\documentclass[11pt]{amsart}
\usepackage[ruled]{algorithm2e} 
\usepackage[sort]{natbib}
\usepackage{amsmath}
\usepackage{amssymb}
\usepackage{amsthm}
\usepackage{graphicx}
\usepackage{tikz}
\usepackage{xcolor}
\usepackage{float}
\usepackage{wrapfig}
\usepackage{hyperref}
\allowdisplaybreaks
\usepackage{bookmark}
\usepackage{enumitem}  

\newtheorem{theorem}{Theorem}
\newtheorem*{theorem*}{Theorem}

\newtheorem*{conjecture*}{Conjecture}

\newtheorem{lemma}[theorem]{Lemma}
\newtheorem{observation}[theorem]{Observation}
\newtheorem{proposition}[theorem]{Proposition}
\newtheorem{definition}[theorem]{Definition}
\newtheorem{claim}[theorem]{Claim}
\newtheorem*{claim*}{Claim}

\renewcommand{\epsilon}{\varepsilon}

\newcommand{\quotebox}[1]
{\medbreak
\fcolorbox{white}{blue!15!gray!15}
{\begin{minipage}{0.95\linewidth}
{\emph{#1}}
\end{minipage}}
\medbreak}

\newcommand{\pparagraph}[1]{
\vspace{0.13in}\noindent{\textbf{\boldmath #1}}~}  

\newcommand{\cal}[1]{\mathcal{#1}}

\newcommand{\E}{\mathbb{E}}
\renewcommand{\Pr}{\mathbb{P}}

\newcommand{\defn}{\textbf}
\newif\ifnotes
\notestrue
\ifnotes
\newcommand{\shikha}[1]{\textcolor{blue}{{\footnotesize #1}}\marginpar{\raggedright\tiny \textcolor{blue}{Shikha}}}
\else
\newcommand{\shikha}[1]{}
\fi

\newif\ifnotes
\notestrue
\ifnotes
\newcommand{\aditya}[1]{\textcolor{purple}{{\footnotesize #1}}\marginpar{\raggedright\tiny \textcolor{purple}{Aditya}}}
\else
\newcommand{\aditya}[1]{}
\fi
\allowdisplaybreaks

\usepackage[margin=1in]{geometry}
\usepackage{bookmark}

\title{Unbalanced Random Matching Markets with Partial Preferences}

\author{Aditya Potukuchi}
\address[AP]{Department of Electrical Engineering and Computer Science, York University, Toronto, ON M3J 1P3, Canada}
\email{apotu@yorku.ca}

\author{Shikha Singh}
\address[SS]{Department of Computer Science, Williams College, Williamstown, MA 01267, USA}
\email{shikha@cs.williams.edu}

\begin{document}

\maketitle


\begin{abstract}

Properties of stable matchings in the popular random-matching-market model have been studied for over 50 years. In a random matching market, each agent has complete preferences drawn uniformly and independently at random.  Wilson (1972), Knuth (1976) and Pittel (1989) proved that in balanced random matching markets, the proposers are matched to  their $\ln n$th choice on average.  In this paper, we consider markets where agents have partial (truncated) preferences, that is, the proposers only rank their top $d$ partners. Despite the long history of the problem, the following fundamental question remained unanswered: \emph{what is the smallest value of $d$ that results in a perfect stable matching with high probability?}  In this paper, we answer this question exactly---we prove that a degree of $\ln^2 n$ is necessary and sufficient.  That is, we show that if $d < (1-\epsilon) \ln^2 n$ then no stable matching is perfect and if $d > (1+ \epsilon) \ln^2 n$, then every stable matching is perfect with high probability. This settles a recent conjecture by Kanoria, Min and Qian (2021).

We generalize this threshold for unbalanced markets: we consider a matching market with $n$ agents on the shorter side and $n(\alpha+1)$ agents on the longer side. We show that for markets with $\alpha =o(1)$, the sharp threshold characterizing the existence of perfect stable matching occurs when $d$ is $\ln n \cdot \ln \left(\frac{1 + \alpha}{\alpha + (1/n(\alpha+1))} \right)$.

Finally, we extend the line of work studying the effect of imbalance on the expected rank of the proposers (termed the ``stark effect of competition''). We establish the regime in unbalanced markets that forces this stark effect to take shape in markets with partial preferences.

\end{abstract}

\section{Introduction}\label{sec:intro}

The stable matching problem is a classic problem in computer science, economics and mathematics.  In the standard and perhaps antiquated description of the problem as the ``stable marriage problem'', the two sides of the market are men and women, each with complete preferences over each other.  A matching is considered to be \defn{stable} if no pair of agents would prefer to break off from their current match and match with each other instead. Stability has proved to be an extremely robust notion in the success of centralized two-sided matching markets~\cite{roth2002economist}.  The first stable matching algorithm proposed by~\cite{gale1962college} has since been used in many real-life matching markets such as matching doctors to hospitals, students to schools, and applicants to jobs; see for example~\cite{roth1999redesign,abdulkadirouglu2005new,abdulkadirouglu2005boston,correa2019school}.

Initiated by~\cite{wilson1972analysis,knuth1976mariages}, there has over 50 years of fascinating research to understand the properties of stable matchings in random-matching markets, where the preferences of the agents are generated uniformly and independently at random.

\subsection{Our Results} 
In this paper, we study a fundamental and natural question about random matching markets with partial preference lists: 
\quotebox{If each agent ranks only their top $d$ choices (uniformly and independently at random), how big does $d$ have to be to ensure that all agents are matched?} 
We consider a two-sided market, with potentially unequal number of agents on each side, and say that a stable matching is \defn{perfect} if it matches all agents on the short side.

As our main result, we tightly characterize the existence of perfect stable matchings.

\begin{theorem*}(Informal)
Given a stable matching market with $n+k$ candidates, $n$ jobs and uniform-random preference lists of length $d$, there exists a \defn{threshold} $d_0$ such that for any (constant) $\epsilon > 0$, 
\begin{itemize}
\item if $d > (1 + \epsilon)d_0$, then all stable matchings are {perfect} with high probability, and
\item if $d < (1 - \epsilon)d_0$, then no stable matching is perfect  with high probability.
\end{itemize}
Moreover, we determine $d_0$ asymptotically:

\begin{enumerate}
    \item\label{part:balanced} for \defn{\em balanced markets} with $k = 0$,~~ $d_0 = \ln^2 n$
    \item\label{part:unbalanced} for \defn{\em unbalanced markets} with $k = o(n)$,~~ $d_0 = \ln n \ln \frac{1 +(k/n)}{(k/n) + 1/(n+k)}$
\end{enumerate}
\end{theorem*}

For balanced markets, the above result settles a recent conjecture by~\cite{kanoria2021matching} who prove an weaker bound---they show $d = o(\ln^2 n)$ is not sufficient and $d = \Omega(\ln^2 n)$ is sufficient for a perfect stable matching to exist.  Even for complete preference lists, the best-known bound on the maximum number of proposals any agent has to make is $\Theta(\log^2 n)$ by~\cite{pittel2019likely}.

For unbalanced markets, Theorem~\ref{thm:unbalancedmain} provides a fundamentally new understanding of perfect stable matchings. Prior to this work, the best-known bounds are from~\cite{kanoria2021matching} who show that the threshold $d_0 = \Theta(\ln^ 2 n)$ continues to hold even in the presence of polynomially-small imbalance.  In particular, for imbalance $k = \Omega(n^{1-\epsilon})$, it was unknown whether (a) such a phenomenon of a sharp threshold continues to persist, and (b) how the existence of perfect stable matchings depends on the imbalance $k$ and preference list $d$ (even asymptotically). 

Finally, we also analyze how the competition in the market affects the best-possible stable partner that the candidates can obtain.

 \begin{theorem*}(Informal)
Given a stable matching market with $n+k$ candidates, $n$ jobs and uniform-random preference lists of length $d$, the expected rank of the best-possible stable partner each candidate can obtain is at least $(1-o(1)) d/(\ln \frac{1 +(k/n)}{(k/n) + 1/(n+k)})$.
\end{theorem*}
 Theorem~\ref{thm:ranklower} captures the intuitive property that the more the competition between the $n+k$ candidates for the $n$ jobs, the worse their stable outcome (as a function of the imbalance $k$).  Thus, it generalizes part of the results in the line of work studying the ``stark effect of competition'' in unbalanced matching markets~\cite{ashlagi2017unbalanced, cai2022short, kanoria2021matching}.

\subsection{Background on Stable Matchings}
To provide some context for our results, we first review the classic \defn{deferred-acceptance (DA) algorithm}  by~\cite{mcvitie1970stable, gale1962college} to find a stable matching.  
Given the preference lists of each candidate and job, consider the following procedure in which the candidates ``propose'' to jobs on their list. At each point, a candidate without a potential match (chosen arbitrarily) proposes to the highest job on their list that they have not been rejected by yet. On receiving a proposal from a candidate $c$, a job $j$ becomes potentially matched to $c$ if $c$ is the highest ranked among the candidates that have proposed to $j$ so far, otherwise $j$ rejects $c$. The process goes on until every candidate has either been rejected by every job on their list, or has a potential match. At the end, the potential matches are final.  

The DA algorithm has some interesting properties that those unfamiliar may find surprising: (a) it always find a stable matching, and (b) this stable matching results in the best-possible stable partner for all candidates, and the worst-possible stable partner for the jobs.  Beyond algorithmic application, these properties also make the DA algorithm an important tool in the study of random instances. Notably, some of the early results by \cite{wilson1972analysis, knuth1976mariages, pittel1989average} on balanced matching markets with uniform-random and complete preference lists show that each candidate on average is matched to their $\ln n$th ranked job and each job is matched to their $(n/\ln n)$th ranked candidate. 

The DA algorithm readily generalizes to partial preferences, which were first studied by~\cite{gusfield1989stable}.  Notice that if the preference list length $d$ is too small, some candidates/jobs may remain unmatched.  Remarkably, the set of unmatched agents stays the same in every stable matching (this is referred to as the Lone Wolf Theorem~\cite{mcvitie1971stable}).  Thus, if a perfect stable matching exists, then all stable matchings are perfect.

\subsection{Model and Formal Statement of Results}
Before we proceed, we formally define the model. Consider a two-sided matching instance $M_{n, d, \alpha}$ with $n(\alpha+1)$ candidates and $n$ jobs, where $\alpha \geq 0$. For balanced markets with $\alpha = 0$, we omit it from the instance and write $M_{n,d}$.

An edge $(c, j)$ between a candidate $c$ and job $j$ occurs independently with probability $d/n$  (thus the resulting graph is a bipartite Erd\"{o}s-R\'{e}nyi graph $G_{n(\alpha+1),n,d/n}$ graph). 
Each candidate and job have random ranking over their neighbors, chosen uniformly and independently. So, $M_{n,n}$ is just the random stable matching instance with complete preference lists.\footnote{We abuse notation and use $M_{n,d}$ to refer both to the stable matching instance (graph and rankings) and the graph itself.}

We refer to $d$ as the \defn{degree} of the underlying random graph.   One may easily deduce that for a perfect stable matching, $d$ needs to be at least $(1-o(1))\ln n$, since otherwise, $M_{n,d}$ has no perfect matching with high probability. It turns out that the critical threshold for perfect stable matchings is a bit higher, and is $\ln^2 n$. 

This is the formal statement of our first main theorem.

\begin{theorem}\label{thm:balancedmain}
For each $\epsilon > 0$,
\begin{enumerate}
\item if $d \geq (1 + \epsilon)\log^2 n$, then w.h.p., all stable matchings in $M_{n,d}$ are perfect, and
\item if $d \leq (1 - \epsilon)\log^2 n$, then w.h.p., no stable matching in $M_{n,d}$ is perfect.
\end{enumerate}
\end{theorem}

The DA algorithm also naturally motivates another distribution over random stable matching instances, where each candidate chooses a set of $d$ (chosen in advance) jobs along with a uniform ranking independently, and each job ranks the candidates uniformly and independently. Let us use $M_{n,d}^{\mathcal{C}}$ to describe this model. This model of \emph{partial} preferences ($d<n$) was first considered by~\cite{immorlica2005marriage, immorlica2015incentives}, who assume that the proposing side has preference lists consisting of only $O(1)$ jobs (drawn independently at random using an arbitrary distribution).  They show that the number of candidates with more than one stable partner makes up a vanishingly small fraction. The study of $M_{n,d}^{\mathcal{C}}$ was revisited recently by~\cite{kanoria2021matching}, who showed that if $d = o(\ln^2 n)$, then w.h.p. no stable matching is perfect and if $d = \Omega(\ln^2 n)$, then w.h.p. every stable matching is perfect. Empirically, they observe that the threshold is about $\ln ^2 n$ for reasonable values of $n$.

It turns out that these different models are closely related to each other in the sense that some results (particularly, the type we are interested in) often translate between them when $d \gg \ln n$. Indeed, we prove that in this regime, a random instance from one of these models can be ``sandwiched'' using instances of another model that are close enough; see Lemma~\ref{lem:sandwich}.

We now formally state our general theorem on threshold for perfect stable matching.

\begin{theorem}\label{thm:unbalancedmain}
Consider a stable matching instance $M_{n, d, \alpha}$ where $0 \leq \alpha$ and $\alpha = o(1)$ and $1 \leq d \leq n$.  Then, for each $\epsilon >0$, 
\begin{enumerate}
    \item\label{part:unbalancedlower} if $d > (1 + \epsilon) \ln n \cdot \ln \left(\frac{1 + \alpha}{\alpha + (1/n(1+\alpha))} \right)$, then w.h.p, all stable matchings are perfect, and

    \item\label{part:unbalancedupper} if $d < (1- \epsilon) \ln n \cdot \ln \left(\frac{1 + \alpha}{\alpha + (1/n(1+\alpha))} \right)$, then w.h.p, no stable matching is perfect.
\end{enumerate}

\end{theorem}

\pparagraph{Competition in unbalanced matching markets.}  A particularly remarkable result by~\cite{ashlagi2017unbalanced} reveals the striking effect of even a small imbalance on random stable matching markets.  They consider random markets with $n+1$ candidates and $n$ jobs and complete preference lists---that is, a stable matching instance $M_{n, n, 1/n}$.  They show that on average each candidate is matched to their $\Omega(n/\ln n)$th ranked job, a stark drop from their $O(\ln n)$th ranked job in the balanced setting.\footnote{Part of this result was later simplified by~\cite{cai2022short}.} 
 Likewise, each job get's its $O(\ln n)$th ranked candidate, a stark rise from the $\Omega(n/\ln n)$ in the balanced case.  Thus, the huge advantage of being on the proposing side (over the non-proposing side) is lost as soon as there is {\em any} competition.  This intriguing discovery further motivates our study of stable matchings in a general $M_{n,d,\alpha}$ setting.

The relationship between imbalance $\alpha$ and degree $d$ was recently studied by~\cite{kanoria2021matching}.  The focus of their work was to study 
the effect of competition as the degree $d$ becomes smaller.  They study stable matching instances $M_{n, d, \alpha}^{\cal{C}}$ with \emph{polynomially small imbalance}, that is, $1/n \leq \alpha \leq 1/n^{\epsilon}$, and show the following dichotomy:  (a) when $d = o(\ln ^2 n)$, there is no perfect matching and both candidates and jobs are matched to the $\sqrt{d}(1+ o(1))$-ranked partner on average and thus there is no effect of competition, and (b) when $d = \Omega(\ln^2 n)$, then all stable matchings are perfect and there is a stark effect of competition---in a candidate-proposing DA, the candidates are matched to their $(d/\ln n)$th ranked job and the jobs to their $\ln n$th ranked candidate.

Our results complements their findings. 
 Theorem~\ref{thm:unbalancedmain} determines the exact threshold for the existence of perfect stable matchings, and, as a function of imbalance $\alpha$ that can be much larger than $1/n^{\epsilon}$.  Finally, Theorem~\ref{thm:ranklower} below improves the lower bound of $(d/\ln n)$ on the expected rank of the candidates in a candidate-proposing DA, and shows how it degrades as the competition grows.
 
 \begin{theorem}\label{thm:ranklower}
Consider a stable matching instance $M_{n, d, \alpha}^{\mathcal{C}}$ where $0 \leq \alpha$ and $\alpha = o(1)$ and $1 \leq d \leq n$. Then, the expected rank of the candidates in the candidate-proposing DA is lower bounded by 
\[(1-o(1)) \frac{d}{\ln \left( \frac{1+\alpha}{\alpha + 1/n(1+\alpha)}\right)}.\]
\end{theorem}

A brief, but rather curious, technical remark: the methods of~\cite{kanoria2021matching} prove a lower bound of $\sqrt{d}(1+o(1))$ (with imbalance $\alpha$ being a lower-order term). This bound is better than the one in Theorem~\ref{thm:ranklower} when $\alpha \ll e^{-\sqrt{d}}$. Interestingly enough, the regime that they work in, $\alpha = n^{-\epsilon}$ and $d = o(\log^2 n)$, reflects this and the competition in this regime plays a negligible role. It would be a natural guess that $\alpha \approx e^{-\sqrt{d}}$ is where the ``stark effect of competition'' phenomenon starts to take shape.  So, we conjecture that the bound in Theorem~\ref{thm:ranklower} is tight when $\alpha \gg e^{-\sqrt{d}}$.

\subsection{Technical Contributions} 
We develop the following technical tools in our analysis.

 First, we compare the proposals in the DA algorithm to a balls-in-bins process. This comparison itself is not new: analyzing the DA as a balls-and-bins or coupon-collector game dates back to the original papers by~\cite{wilson1972analysis, knuth1976mariages}. We add to this approach by providing some general bounds using simple methods that could simplify similar proofs in the future.

 Second, to establish whether or not a stable matching exists, we focus on which of the two competing forces in the DA algorithm occurs first:  each of the $n$ jobs receiving a proposal or some candidate exhausting all $d$ of their proposals. We analyze this competition by defining a simple probabilistic game on the ``rejection chains'' in the DA algorithm. \cite{kanoria2021matching} also analyzes rejection chains, but at a conceptual level, understanding the competing forces seems to offer a fundamentally different approach. For instance, the analysis in~\cite{kanoria2021matching} requires a bound on the total number of proposals made in the DA algorithm. The aforementioned game lets us circumvent this and isolate the event of interest (the existence of a perfect stable matching).

Finally, we study three closely related random models: the symmetric case $M_{n, d, \alpha}$ which is arguably the most natural to state as it is independent of which side of the market proposes,  
 and the asymmetric models $M_{n, d, \alpha}^{\cal{C}}$ and $M_{n, d, \alpha}^{\cal{J}}$ in which the candidates and jobs choose preference lists of length $d$ respectively.  Switching between these models helps analyze the problem from both sides: in particular, we analyze the candidate-proposing DA on $M_{n, d, \alpha}^{\cal{C}}$ and the job-proposing DA on $M_{n, d, \alpha}^{\cal{J}}$.  We combine these results to obtain a bound for $M_{n, d, \alpha}$  using a ``sandwiching lemma'' (Lemma~\ref{lem:sandwich}).  We believe that this approach of analyzing the DA on different random stable-matching instances provides new insights as past work has only studied the instance $M_{n, d, \alpha}^{\cal{C}}$.

\subsection{Additional Related Work}\label{sec:related}
The original stable matching algorithm by~\cite{gale1962college} was extended to the imbalanced case by~\cite{mcvitie1970stable}, and to the case of partial preferences by~\cite{gusfield1989stable}. Stable matchings have since been extensively studied over 60 years with several books on the topic;  we refer to the readers to~\cite{manlove2013algorithmics}.

Here we present closely related results on random-matching markets. \cite{wilson1972analysis, knuth1976mariages}  reduce the problem of running deferred acceptance in random balanced matching markets with complete lists, $M_{n,n}$, to the coupon-collector problem and use it to compute its average complexity.  This technique is used in various works (including this paper), e.g.~\cite{kanoria2021matching, ashlagi_et_al:LIPIcs.ITCS.2021.46,ashlagi2017unbalanced}. \cite{Wilson1972AnAO} showed that the expected number of proposals in the candidate-proposing DA is at most $n H_n$. \cite{knuth1976mariages} improved this upper bound to $(n-1) H_n +1$ and proved a matching lower bound of $n H_n - O(\ln^4 n)$.  Thus, the candidates side are matched to their $\ln n$th ranked partner in expectation. In contrast, \cite{pittel1989average} showed that the receiving side is a lot worse and is matched to their $(n/\ln n)$th partner in expectation.

 The algorithm to enumerate all stable partners of an agent (by breaking matches and continually running deferred acceptance) was suggested by McVittie~\cite{mcvitie1971stable} and Gusfield~\cite{gusfield1987three}.  In this paper, we use the simplified version by~\cite{knuth1990stable}.  This technique has also often used to show that each agent has a unique stable partner~\cite{immorlica2005marriage, ashlagi2017unbalanced}, the phenomenon referred to as ``core convergence''.

 Non-uniform preference distributions have been studied, such as tiered preferences~\cite{ashlagi_et_al:LIPIcs.ITCS.2021.46} or correlated preferences~\cite{gimbert2021two}. Partial preference lists based on public scores was recently studied by~\cite{agarwal2023stable}. 
 
 For a recent survey on the history and results on random matching markets, see~\cite{mauras2021analysis}.

\subsection{Organization} We describe the model, the deferred acceptance algorithm, as well the main building blocks of our analyses---the probability game---in Section~\ref{sec:prelim}.
We give the the tight upper and lower bounds on the threshold in Section~\ref{sec:threshold}.  Finally, we analyze the expected rank in Section~\ref{sec:expected}.  

\section{Preliminaries}\label{sec:prelim}

In this section, we formally define our model, the stable matching algorithm used and other definition and notation used in the rest of the paper.

\pparagraph{Stable matchings.}  Consider a stable matching instance $M$ with $n(\alpha+1)$ candidates $\mathcal{C}$ and $n$ jobs $\mathcal{J}$. Without loss of generality assume that $\alpha \geq 0$. Each candidate $c \in \mathcal{C}$ has a (partial) preference list of length $d_c^M$ over a subset of jobs ranking them from the most to lease preferred.  Similarly, each job $j \in \mathcal{J}$ has a (partial) preference list of length $d_j^M$ over a subset of candidates ranking them from most to lease preferred. Consider the undirected bipartite graph $G^M$ on $\mathcal{C} \cup \mathcal{J}$ where the edge $(c, j)$ is present if $c$ appears on $j$'s preference list and vice versa. Let $L_c^M= (j_1,\ldots,j_{d_c^M})$ to denote the list of the candidate $c$ in $M$. We define $L_j^M$ for jobs analogously. 

For stable matching instances $M_1$ and $M_2$ on $\mathcal{C} \cup \mathcal{J}$, we say that $M_1 \subseteq_{\mathcal{C}}M_2$ if for each candidate $c \in C$, we have that $d_{c}^{M_1} \leq d_c^{M_2}$ and the first $d_{c}^{M_1}$ elements and their order in $L_c^{M_2}$ is the same as in $L_c^{M_1}$. Moreover, for every job $j$, the relative ranking of the candidates in $L_j^{M_1}$ remains the same. Similarly, we say that $M_1\subseteq_{\mathcal{J}}M_2$ if the same holds for every $j \in \mathcal{J}$.

A \defn{matching} $\mu$ is a set of vertex-disjoint edges in $G^M$.  A matching $\mu$ is \defn{perfect} if each job $j$ is an endpoint of some edge in $\mu$.
For simplicity, we let $\mu(c)$ denote a candidate $c$'s match in $\mu$ (let $\mu(c) = \emptyset$ if $c$ is unmatched).  We define $\mu(j)$ similarly for a job $j$.
A matching $\mu$ is \defn{stable} if there exists no \defn{blocking pair} with respect to $\mu$.  A pair $(c, j)$ where $c \in \mathcal{C}$ and $j \in \mathcal{J}$ is \defn{a blocking pair} with respect to the matching $\mu$ if any one of the following conditions hold:  (a) $c$ prefers $j$ to $\mu(c)$  and $j$ prefers $c$ to $\mu(j)$, (b) $c$ prefers $\emptyset$ to $\mu(c)$, and (c)  $j$ prefers $\emptyset$ to $\mu(j)$.

\pparagraph{Deferred-acceptance algorithm.}
We describe the \defn{candidate-proposing deferred-acceptance} (CPDA) algorithm in Algorithm~\ref{alg:cpda}.  The job-proposing DA (JPDA) is analogous.  

Let $\cal{U}$ be the set of unmatched candidates.  Let $\cal{R}$ be the set of candidates who have \defn{exhausted} their preference list, that is, have been rejected by all jobs on their preference list.  

\begin{algorithm}
\caption{Candidate Proposing Deferred-Acceptance (CPDA) Algorithm }\label{alg:cpda}
Fix an ordering over $\cal{C}$;   Initialize $\cal{U} \gets \cal{C}$ and $\cal{R} \gets \emptyset$ and $\mu(j) \gets \emptyset$ for all $j \in \cal{J}$

\While{$\cal{U} \setminus \cal{R} \neq \emptyset$}{
    Pick a candidate $c \in \cal{U}$ with the lowest index  

    \If{$c \notin \cal{R}$}{

    Let $j$ be the most preferred job on $c$'s list that $c$ has not yet proposed to
    
    \If{\emph{$j$ is the last job on $c$'s list}}
        {add $c$ to $\cal{R}$}
        
    $c$ \defn{proposes} to $j$;  Let $c' = \mu(j)$
    
    \If{\emph{$j$ prefers $c$ to $c'$}}
          {
          $j$ \defn{rejects} $c'$
          
          \If{$c' \neq \emptyset$}
            {add $c'$ to $\cal{U}$}
          $j$ \defn{accepts} $c$;
    
       Set $\mu(j) = c$ and remove $c$ from $\cal{U}$
          }
    }
  }
\end{algorithm}

We use the following properties of the CPDA algorithm.
\begin{theorem}\label{thm:DAfacts}{\cite{gale1962college,mcvitie1970stable}}
The CPDA algorithm has the following properties:
\begin{enumerate}
\item It outputs a stable matching $\mu^*$.  The matching $\mu*$ is \defn{candidate optimal}, that is, each candidate is matched to their best-possible stable match.  Moreover, $\mu^*$ is \defn{job pessimal}, that is, each job is matched to their worst-possible stable match. 
\item\label{label:order} The matching $\mu^*$ is independent of the order of the proposals.  Moreover, the {same proposals} are made during the algorithm, regardless of which candidate is chosen to propose at each step.
\item\label{label:lonewolf} (\textbf{Lone Wolf Theorem}) The candidates and jobs that are unmatched in $\mu^*$ do not have any possible stable partner. That is, the set of unmatched agents is the same across all stable matchings.
\end{enumerate}
\end{theorem}

\pparagraph{Algorithm terminology.}  
We define the current \defn{time step} of the CPDA algorithm as the number of proposals made so far. Since the next proposal is always from a currently unmatched candidate with the lowest index, a candidate keeps proposing until either they are matched or exhaust their preference list. 
Moreover, if the $t$th proposal results in a rejection of a candidate $c$ (either because $c$ made the proposal, or because the $t$th proposal causes a previously-matched $c$ to become unmatched) the next proposal is always from $c$.  This sequence of proposals is termed as the \defn{rejection chain} in the literature.   We use rejection chains in our analysis and define it more formally below.

\begin{definition}(Rejection Chain)\label{def:rejectionchain}
Fix any time step $t$ in CPDA (Algorithm~\ref{alg:cpda}) and let $c$ be the next candidate in line to propose. The rejection chain $R(t)$ starting at time step $t$ is empty if $c$ has been rejected by all jobs on their preference list.  Otherwise, let $j$ the job $c$ proposes to next.
\begin{itemize}
    \item If  $j$ is currently unmatched, then $R(t) = (c, j, \emph{accept})$.
    \item Otherwise, if $j$ prefers $\mu(j)$ to $c$, then $R(t) = (c, j, \emph{reject}) \circ R(t+1)$.  
    \item Otherwise, if $j$ prefers $c$ to $\mu(j)$, then
    $R(t) = ((c, j, \emph{accept}), (c', j, \emph{reject})) \circ R(t+1)$.
\end{itemize}
\end{definition}

We use the following two observations about rejection chains in our analysis.
\begin{observation}\label{obs:rejectchain}
Consider the rejection chain $R(t)$ at time step $t$ defined in Definition~\ref{def:rejectionchain}.
\begin{enumerate}
    \item\label{obs:part1} $R(t)$ terminates when either previously unmatched job receives a proposal and thus will remain matched till the end, or a candidate is rejected by all jobs on their preference list and thus will remain unmatched till the end.
    \item\label{obs:part2} In a rejection chain, if $d$ contiguous proposals result in a rejection, then some candidate must have been rejected $d$ times and must remain unmatched.  
\end{enumerate}
\end{observation}
Note that Part~(\ref{obs:part2}) of the above observation was also used by~\cite{kanoria2021matching}.

\pparagraph{Deferred Acceptance on Random Instances.}
We recall the random stable matching instance $M_{n, d, \alpha}$ defined in Section~\ref{sec:intro}.
In this instance, there are $m = n(\alpha+1)$ candidates and $n$ jobs.  The underlying graph is a Erd\"{o}s-R\'{e}nyi graph $G_{n(\alpha+1), n, d/n}$. Thus, an edge $(c, j)$ between a candidate $c$ and job $j$ occurs with independent probability $d/n$. Moreover, each candidate and job have a random ranking over their neighbors, chosen uniformly and independently.  The stable matching instance $M_{n, d, \alpha}$ refers to the graph $G_{n(\alpha+1), n, d/n}$ and the preferences over the edges.  

Let us define two other related models of random stable matching instances. 
 Let us define $M^{\mathcal{C}}_{n,\alpha,d}$ to denote a stable matching instance with $n$ jobs, $n(1+\alpha)$ candidates, and each candidates chooses a preference list of $d$ jobs uniformly and independently. This is the model in~\cite{kanoria2021matching}.
 
 Similarly, we also define $M^{\mathcal{J}}_{n,\alpha,d}$ where each job chooses a preference list of $d$ candidates uniformly and independently.

We analyze the CPDA algorithm on $M^{\mathcal{C}}_{n,\alpha,d}$ and the JPDA algorithm on $M^{\mathcal{J}}_{n,\alpha,d}$.  To obtain the final result on the symmetric stable matching instance  $M_{n, d, \alpha}$, we relate all three random instances using the following ``sandwiching'' lemma.

\begin{lemma}(Sandwiching Lemma)\label{lem:sandwich}
For every $\delta > 0$, there is a joint distribution
$(M^{\mathcal{C}}_{n,\alpha,d(1-\delta)},$ $M_{n,\alpha,d},M^{\mathcal{C}}_{n,\alpha,d(1 + \delta)})$ such that with probability at least $1 - O\left(n\exp\{-\delta^2 d/3\}\right)$, we have 
\[
M^C_{n,\alpha,d(1-\delta)} \subseteq_{\mathcal{C}} M_{n,\alpha,d} \subseteq_{\mathcal{C}} M^C_{n,\alpha,d(1 + \delta)}.
\]
The same claim also holds for $(M^{\mathcal{J}}_{n,\alpha, d(1-\delta)}, M_{n,\alpha,d(1+\alpha)},M^{\mathcal{J}}_{n,\alpha, d(1+\delta)})$.
\end{lemma}

Since we (eventually) analyze the DA algorithm on a random instance $M_{n,\alpha, d}^{\mathcal{C}}$, we reveal the next job $j$ in the list of candidate $c$ when they are called up to propose, uniformly among all jobs that have not already rejected $c$. We also reveal $c$'s relative ranking among all the proposals received by $j$ (including from $c$) when $c$ proposes to $j$.

\subsection{A Probability Game}
\label{sec:game}

Define a multi-round probabilistic game as follows. In each round, there are three events that can occur: $G$ (Good), $B$ (Bad), and $N$ (Neutral).  Given an input parameter $d$,  for each round, conditioned on the events in previous rounds, define
$\Pr(G) = p_G$, $\Pr(B) \geq p_B$ and $\Pr(N) = p_N$.  The game is \defn{won} if $G$ occurs.  The game is \defn{lost} if $B$ occurs $\delta$ times \emph{in a row}.  

Let $\Pr_{\text{win}}(p_G, p_B)$ denote the probability of winning this game.  Then, we prove the following.

\begin{lemma}\label{lem:probgame}
%
\begin{equation*}
  \Pr_{\text{\em win}}(p_G, p_B) 
  \leq \frac{p_G}{p_B^d}
\end{equation*}
\end{lemma}

\begin{proof}
Suppose the game is played for $i$ rounds, and let $S = (s_1, s_2, \ldots, s_i)$ denote the random sequence of events that occur.  Consider a new game that is the same as the original except now $\Pr(B) = p_B$, and let $S' = (s_1', s_2', \ldots, s_i')$ denote the random sequence of events that occur when the new game is played for $i$ rounds.  
Let $\Pr_{\text{win}'}$ be the winning probability in the new game.  Through a coupling argument, we first show that 
$\Pr_{\text{win}} \leq \Pr_{\text{win}'}$.

Given the sequence $S'$, we create a sequence $S$ as follows.  For $i\geq 0$, as long as the game does not end in a loss at $i$ rounds, do:
\begin{itemize}
    \item if $s_i' = G$, then $s_i = G$,
    \item if $s_i' = B$, then $s_i = B$,
    \item if $s_i' = N$, then $s_i = N$ with probability $\frac{1 - p_G - p_B^i}{1 - p_G - p_B}$ and $s_i = B$ otherwise.  
    Here $p_B^i = \Pr(s_i = B~|~ s_{i-1}, \ldots, s_1) \geq p_B$.  
\end{itemize}
Note that whenever $S'$ ends in a loss (event $B$ occurs $d$ times in a row), then $S$ also ends in a loss.   Thus, $\Pr_{\text{win}} \leq \Pr_{\text{win}'}$.  

Finally, we prove finish the proof of the lemma by using the following recurrence.

\begin{align*}
\Pr_{\text{win}'} &= p_G + p_N \cdot \Pr_{\text{win}'} + p_B \cdot p_N \cdot \Pr_{\text{win}'} + p_B^2 \cdot p_N \cdot \Pr_{\text{win}'}  + \ldots +  p_B^{d-1} \cdot p_N \cdot \Pr_{\text{win}'}\\
&= p_G + p_N \cdot \Pr_{\text{win}'} (1  + p_B + p_B^2 + \ldots + p_B^{d-1}) = p_G + p_N \cdot 
\Pr_{\text{win}'} \left( \frac{1- p_B^d}{1-p_B}\right)\\
%
\Pr_{\text{win}'} &= \frac{p_G(1-p_B)}{1 - p_B - p_N + p_N \cdot p_B^d} = \frac{p_G(1-p_B)}{p_G + p_N \cdot p_B^d} \\
&= \frac{p_G}{p_B^d} \cdot \frac{(1-p_B)}{(p_G/p_B^d+p_N) }
\leq \frac{p_G}{p_B^d}\qedhere
\end{align*}

\end{proof}

\pparagraph{Notation.}  
We say an event $\mathcal{E}$ occurs \defn{with high probability} if $\mathbb{P}(\mathcal{E}) = 1-o(1)$.

\section{Proof of Theorem~\ref{thm:unbalancedmain}}\label{sec:threshold}
In this section, we prove Theorem~\ref{thm:unbalancedmain}.  That is, we prove tight upper and lower bounds on the degree $d$ that determines whether or not a stable matching exists in a random matching market with imbalance $\alpha$ and average degree $d$. As alluded to earlier, we prove the lower bound on $M_{n,\alpha,d(1+o(1))}^{\mathcal{C}}$, and upper bound on $M_{n,\alpha,d(1-o(1))}^{\mathcal{J}}$.

\begin{theorem}\label{thm:threshlower}
For any $\gamma>0$, if $d < (1-\gamma) \log n \cdot \log \left( \frac {1 + \alpha}{\alpha + (1/n(1+\alpha))} \right)$, then 
\[
\mathbb{P}\left(M_{n, \alpha, d}^{\cal{C}}~\emph{has no perfect stable matching}\right) \geq 1 - (\alpha + 1/n)^{-\Omega(\gamma)}.
\]
\end{theorem}

\begin{theorem}\label{thm:threshupper}
 For any $\gamma > 0$, if $d > (1+\gamma) \log n \cdot \log \left( \frac {1+\alpha}{\alpha + (1/n(1+\alpha))} \right)$, then 
 \[\mathbb{P}(M_{n, \alpha, d}^{\cal{J}} \emph{ has a perfect stable matching}) \geq 1 -{n^{-\Omega(\gamma)}}\]
\end{theorem}

We prove Theorems~\ref{thm:threshlower} and~\ref{thm:threshupper} in Section~\ref{sec:lower} and Section~\ref{sec:upper} respectively.  Using them, we finish the proof of Theorem~\ref{thm:unbalancedmain} below.

\begin{proof}[Proof of Theorem~\ref{thm:unbalancedmain}]
Let us denote $d_0 := \log n\cdot \log\left(\frac{1 + \alpha}{\alpha + 1/n(1+\alpha)}\right) = \omega(\log n)$. We have for $\alpha = o(1)$ that $d_0 = \omega(\log n)$.

Denote $M = M_{n,\alpha,d_0(1 - \epsilon)}$, and $M^{\mathcal{C}} = M^{\mathcal{C}}_{n,\alpha,d_0(1- \epsilon/2)}$. For the proof of part~\ref{part:unbalancedlower}, it suffices to prove the same statement for $M^{\mathcal{C}}$. Indeed, since by Lemma~\ref{lem:sandwich}, there is a joint distribution $(M,M^{\mathcal{C}})$ such that $M\subseteq_{\mathcal{C}} M^{\mathcal{C}}$ with probability at least $1 - n\exp\{-\Omega(\epsilon^2 d_0)\})$. By part~\ref{label:order} of Theorem~\ref{thm:DAfacts}, we have that if the CPDA algorithm on $M^{\mathcal{C}}$ does not give a perfect stable matching, then it does not give one even in $M$.

The proof of part~\ref{part:unbalancedupper} proceeds in a similar fashion, using $M^{\mathcal{J}}_{n,\alpha,(1+\alpha)d_0(1 + \epsilon/2)}$ .
\end{proof}

\subsection{High-Level Overview}
We first describe the high-level idea of the upper and lower bound proofs for the balanced case ($\alpha=0$) to give some intuition.  

\pparagraph{Lower bound.}  To show that a perfect stable matching does not exist when the average degree $d < (1- \gamma) \log^2 n$ in $M_{n, d}$, we analyze the proposals when the CPDA algorithm is run on the random instance $M_{n, d}^{\mathcal{C}}$.  Thus, each candidate has a preference list of length $d$. Notice that in the CPDA, as soon as all $n$ jobs receive even a single proposal, the final matching must be perfect.  This is because jobs only upgrade their potential matches as the algorithm proceeds. On the other hand, a candidate $c$ remains unmatched if $c$ has been rejected by all $d$ jobs on their preference list.  Thus, to show that a perfect stable matching does not exist, we show that some candidate exhausts all $d$ of their proposals \emph{before} all $n$ jobs receive a proposal.  We model this using the probability game from Section~\ref{sec:game}.  Finally, we set the probabilities in the game using a careful coupling between CPDA and a balls-in-bins process.

This proof naturally extends to the unbalanced case:  we simply analyze the probability that $\alpha n + 1$ candidates exhaust all $d$ of their proposals before all $n$ jobs receive a proposal.

\pparagraph{Upper bound.} To show that a perfect stable matching exists when the average degree $d > (1+ \gamma) \log^2 n$ in $M_{n, d}$, we analyze the proposals in the JPDA algorithm is run on the random instance $M_{n, d}^{\mathcal{J}}$.  Thus, each job has a preference list of length $d$.
To prove that each job gets matched, we look at all the proposals made by the last job $j$ in JPDA and compute the probability that $j$ gets rejected from all $d$ candidates on its list and show that it is small.  For a single proposal from $j$ to a candidate $c$ to be successful, not only must $c$ accept it but the resulting rejection chain must not knock $j$ off.  We use the probability game to compute the probability of the latter.  

\subsection{Lower Bound: Proof of Theorem~\ref{thm:threshlower}}\label{sec:lower}

\begin{proof}[Proof of Theorem~\ref{thm:threshlower}]

Let $\mathcal{P}$ denote the event that $M_{n, \alpha, d}^{\cal{C}}$ has a perfect stable matching. Recall that in $M_{n, \alpha, d}^{\cal{C}}$, each candidate has a preference list of length $d$.  We consider the CPDA procedure in Algorithm~\ref{alg:cpda} on $M_{n, \alpha, d}^{\cal{C}}$.  We define an \defn{extended process} in the algorithm where there are $n(1+\alpha)$ real candidates and $\infty$ `fake' candidates and $n$ jobs. Recall that we assume that there is some predetermined order on the candidates and the unmatched candidate with the least index is next in line to propose. 

Fix $\kappa = n \ln n - C n$, for $C = O(\gamma \ln(1/(\alpha + 1/n)))$.

Suppose $\kappa$ proposals have been made.  We consider the rejection chain $R(\kappa)$ as defined in Definition~\ref{def:rejectionchain}.  
In particular, consider the proposal made by candidate $\tilde{c}$ at step $\kappa + 1$.  We have the
following cases: (a) the proposal goes to an unmatched job and is accepted, (b) the proposal
goes a matched job and is rejected, and (c) the proposal goes to a matched job, is accepted,
which causes its previous match $c'$ now to be free. 
By Observation~\ref{obs:rejectchain}, $R(\kappa)$ ends when either a proposal goes to an unmatched job, or when some candidate exhausts all their $d$ proposals (and thus remains unmatched).

Let us define the following events and then bound their probabilities. 
\begin{itemize}
    \item Let $\mathcal{M}$ be the event that CPDA finds a (potential) matching of size $n$ before $\alpha n + 1$ candidates have exhausted all $d$ of their proposals.\footnote{Note that a potential matching of size $n$ is found as soon as $n$ jobs have received a proposal from some candidate (not necessarily their final match).}
   \item Let $\mathcal{L}$ be the event that every job gets a proposal from a candidate in at most $\kappa$ rounds.
   \item Let $\mathcal{B}$ be the event that after $\kappa$ rounds, the probability that the next proposal from a candidate is accepted is greater than $\frac{n}{n + \kappa}(1 + \frac \gamma3)$.
  \item Let $\mathcal{E}$ be the event that the number of unmatched jobs at $\kappa$ rounds is at least $e^{2C}$.  
\end{itemize}

Then, 
\begin{equation}\label{eqn:lowerboundmain}
    \mathbb{P}(\mathcal{P}) \leq \mathbb{P}(\mathcal{M}) \leq \mathbb{P}(\mathcal{M} | \overline{\mathcal{B}} \overline{\mathcal{E}}\overline{\mathcal{L}}) + \mathbb{P}(\mathcal{B}) + \mathbb{P}(\mathcal{E}) + \mathbb{P}(\mathcal{L})
\end{equation}

As there are $n (\alpha +1)$ real candidates and $n$ jobs, there is no perfect stable matching if more than $\alpha n$ candidates end up with all $d$ of their proposals rejected before all $n$ jobs receive even a single proposal. We have the following claim whose proof is postponed to Section~\ref{sec:appendix}. 

\begin{claim}\label{claim:probbounds}
The following inequalities hold:
\begin{enumerate}
\item\label{probb} $\mathbb{P}(\mathcal{B}) \leq O(1/\sqrt{n})$.
\item\label{probe} $\mathbb{P}(\mathcal{E}) \leq e^{-C}$.
\item\label{probl} $\mathbb{P}(\mathcal{L}) \leq 2e^{-C/2} + O(1/\sqrt{n})$
\end{enumerate}
\end{claim}

With Claim~\ref{claim:probbounds} in hand, we finish the proof by using the probability game defined in Section~\ref{sec:game}.  

Let $p_B := 1 - \frac{n}{n+\kappa}(1+ \frac \gamma 2)$. Conditioned on $\overline{\cal{B}}$, the next proposal is rejected with probability at least $p_B$. Conditioned on $\overline{\mathcal{E}}$, let $p_G$ be the probability of the next proposal goes to an unmatched job. By Observation~\ref{obs:rejectchain}~(\ref{obs:part1}) the number of unmatched jobs stays the same during all proposals that occur within the rejection chain $R(\kappa)$.  
By Observation~\ref{obs:rejectchain}~(\ref{obs:part2}), any time there are $d$ rejects in a row in $R(\kappa)$, some candidate must have exhausted all $d$ of their proposals. 

Now, we use the probability game from Section~\ref{sec:game} and set $p_G = \frac{e^{2C}}{n}$ and $p_B =p$. 
Using Lemma~\ref{lem:probgame} and the fact that $\mathbb{P}_{\text{win}}(p_G, p_B)$ is the probability that every subsequent rejection chain ends in a new job getting matched, we get the following:

\begin{align*}
 \mathbb{P}_{\text{win}}\left( p_G, p_B\right) & \leq p_G{p_B^{-d}} \\
 & \leq e^{2C}n^{-1}\left(1 - \frac{n}{n + \kappa}(1+ \gamma/3)\right)^{-(1-\gamma)\ln n\ln\left(\frac{1+ \alpha}{\alpha + 1/n(1 + \alpha)}\right)}\\
 & \leq e^{2C}n^{-1}\left( 
 \frac{\alpha + 1/n(1 + \alpha)}{1+\alpha}\right)^{-1+\gamma/2} .
\end{align*}

In the last inequality, we used the fact that $1-x \geq e^{-x -x^2}$, and that $\kappa = \Omega(n\ln n)$. So, with probability at least 
$1- e^{2C} \left( \alpha + \frac{1}{n(1+\alpha)}\right)^{\gamma/4}(1+\alpha)^{1 - \gamma/2}$
the next $n \cdot \left( \alpha + \frac{1}{n(1+\alpha)}\right)^{1- \gamma/4} >\alpha n$
rejection chains end in an additional unmatched candidate. Therefore,
\begin{equation}\label{eqn:lowerboundconditional}
\mathbb{P}(\mathcal{M} | \overline{\mathcal{B}} \overline{\mathcal{E}}\overline{\mathcal{L}}) \leq e^{2C} \left( \alpha + \frac{1}{n(1+\alpha)}\right)^{\gamma/4}(1+\alpha)^{1 - \gamma/2}.
\end{equation}
Finally, substituting $C = O(\gamma \ln(1/(\alpha + 1/n)))$, and plugging in the bounds from Claim~\ref{claim:probbounds} and Inequality~\eqref{eqn:lowerboundconditional} into the expression in Inequality~\eqref{eqn:lowerboundmain}, we get 
\[
\mathbb{P}(\mathcal{P})\leq (\alpha + 1/n)^{-\Omega(\gamma)}.
\]
\end{proof}

\subsection{Upper Bound: Proof of Theorem~\ref{thm:threshupper}}\label{sec:upper}
In particular, we prove the following.

\begin{proof}[Proof of Theorem~\ref{thm:threshupper}]

For convenience, let $m = n(1 + \alpha)$. Suppose $d > (1+\gamma) \log n \cdot \ln \left( \frac {1+\alpha}{\alpha + 1/m} \right)$.  Consider the JPDA algorithm (analogous to Algorithm~\ref{alg:cpda}).  Fix a job $j$, and and let us place it in the last place in our predetermined order for the JPDA algorithm---that is, $j$ is the last job that has not yet proposed to any candidate on their list. Let $\kappa$ be the number of proposals at the moment $j$ starts proposing.  Let (the random candidate) $c$ be $j$'s next proposal and $N_c$ be the random variable denoting the number of proposals $c$ has received so far. We have that $\mathbf{E}[N_c|\kappa] = \kappa/m$. Observe that $j$'s proposal to $c$ is \defn{successful} if the following two events occur:
\begin{enumerate}
    \item $c$ accepts $j$'s proposal, and 
    \item the resulting rejection chain $R(\kappa)$ does not end with $j$ getting unmatched.
\end{enumerate} 

The first event occurs with probability $1/(N_c+1)$.  

To bound the probability of the second event, we analyze the proposals in $R(\tau)$ using the probability game from Section~\ref{sec:game}. In the particular, consider the probability game with parameter $d = 1$. Let $G$ be the event that $j$ gets unmatched because another job $j'$ proposes to $c$ in $R(\tau)$.  Then, $p_G \leq 1/(m(N_c +2))$ because the probability that $c$ is next on $j'$ preference list is at most $1/m$ and $c$ must prefer $j'$ to the $N_c + 1$ proposals it already has received.  Let $B$ be the event that the rejection chain $R(\kappa)$ ends without $j$ getting kicked off from $c$ in $R(\kappa)$.  This occurs if $R(\kappa)$ ends in a new candidate getting matched.  Thus $p_B \geq \frac{\alpha + 1/m}{1+\alpha}$.

Then, $\mathbb{P}_{\text{win}}$ corresponds to the probability that $j$ does not end up being matched to $c$ at the end of $R(\kappa)$, and by Lemma~\ref{lem:probgame} we have,

\begin{align}
    \mathbb{P}_{\text{win}} \leq p_G/p_B \leq \frac{1+\alpha}{(\alpha + 1/m) \cdot m \cdot (N_c+2)} 
\end{align}

Thus, the probability that $j$'s proposal to $c$ is successful is  

\begin{align*}
    &\geq \E \left[ \frac{1}{N_c+1} \left(1 - \frac{1+\alpha}{(\alpha + 1/m) \cdot m \cdot (N_c+2)} \right)\right] \\
    &\geq \E \left[ \frac{1}{N_c+1} \left(1 -  \frac{1}{N_c+2}\right) \right]\\
    &= \E \left[ \frac{1}{N_c+2}\right] \geq \frac{1}{\E\left[N_c+2\right]} = \frac{m}{\E[\kappa] + 2m} \geq \frac{m}{\E[\tau] + 2m}.
\end{align*}

Here $\tau$ denotes the final number of proposals (at the end of JPDA) and $\kappa$ is the number of proposals at the time $j$ started proposing.

Finally, the probability that $j$ remains unmatched at the end of JPDA is the probability that all $d$ of $j$'s proposals fail. This occurs with probability
\begin{align}
    &\leq {\left[ 1 - \frac{m}{\E[\tau] + 2m}\right]}^d \leq \exp\left\{-\frac{m \cdot d}{\E[\tau] +2m}\right\}\nonumber \\ 
    &\leq \exp \left\{ - \frac{d}{2 + \ln\left(\frac{1+\alpha}{\alpha+1/m}\right)}\right\} \leq \frac{1}{n^{1+\gamma/2}}. \label{eq:taubound}
\end{align}
Here, the inequality in step~\ref{eq:taubound} follows from the following claim that is postponed to Section~\ref{sec:appendix}.   

\begin{claim}\label{lem:jpda-balls-bins}
$\E[\tau] \leq \left(1 + \gamma/2\right) \cdot m \cdot \log \left(\frac{1+\alpha}{\alpha + 1/m}\right) + o(1)$
\end{claim}

Theorem~\ref{thm:DAfacts} part~\ref{label:lonewolf} gives us that this holds for any job $j$. Thus the expected number of unmatched jobs is at most $n^{-\gamma/2}$, which gives us the desired bound.

\end{proof}

\section{Effect of Competition: Proof of Theorem~\ref{thm:ranklower}}\label{sec:expected}

In this section, we prove a lower bound on the expected ranks of the candidates and the number of proposals in CPDA in an unbalanced random matching market with imbalance $\alpha$.  In particular, we prove Theorem~\ref{thm:ranklower}.

\begin{proof}[Proof of Theorem~\ref{thm:ranklower}]
Consider the JPDA algorithm on $M_{n, d, \alpha}^{\mathcal{J}}$.  We use the following lemma from past works~\cite{cai2022short, immorlica2005marriage,knuth1976mariages}.

\begin{lemma}\label{lem:PDA*}(~\cite{cai2022short})
Fix a job $j$ and for any preference list $P_j$ of job $j$ and index $i \in [n]$, let $P_j^i$ denote the preference list $P_j$ truncated at index $i$.  Then, job $j$ has a stable partner of rank better
than $i$ if and only if $j$ is matched in CPDA even if $j$ truncates their list after rank $i$.
\end{lemma}

We fix a candidate $c^*$. Let $X$ be the random variable that represents $c^*$'s highest-rank stable partner (among all stable matchings). Let $Y$ be the highest rank offer $c$ receives if in the JPDA algorithm $c$ keeps rejecting all proposals (and we keep running JPDA). Lemma~\ref{lem:PDA*} gives us that $\mathbf{E}[X] = \mathbf{E}[Y]$. Let $Z$ be the random variable representing the number of proposals received by $c^*$ during this algorithm where $c^*$ rejects all proposals.

Let us compare the JPDA process to a random balls-in-bins process. Assume that there are $m = n(1 +\alpha)$ bins (one for each candidate). Suppose balls are being thrown uniformly and independently in the bins. Fix a bin $c^*$ and let $\tilde{Z}$ be the number of balls in $c^*$ at the time that exactly $n + 1$ bins are occupied. Let $\kappa$ be the number of balls thrown at this point. We show that this process {\em stochastically dominates} JPDA as follows.

\begin{lemma}
There is a coupling between the JPDA algorithm and the balls-in-bins process so that $Z \leq \tilde{Z}$, and $\tau \leq \kappa$.
\end{lemma}

\begin{proof}
We abuse notation and use $[m]$ to denote both the set of bins and the set of candidates. Consider the following coupling between the processes of JPDA and balls-in-bins.

Suppose at a point in the JPDA algorithm, a job $j \in [n]$ is required to make a (random) proposal, and so far, $\mathbf{C}_j$ is the set of candidates that $j$ has already proposed to. We keep throwing balls-into-bins uniformly at random until we hit a bin $\mathbf{x} \not\in \mathbf{C}_j$. We assign $\mathbf{x}$ as the next proposal of $j$. Thus the distribution of the next proposal of $j$ is uniform on $[m] \setminus \mathbf{C}_j$. 

Let $\tau$ and $\kappa$ denote the number of proposals made and number of balls thrown in the above algorithm respectively. Fix any $i \in [m]$, and let $Z$ and $\hat{Z}$ denote the number of proposals received by the man $i$ and the number of balls in the bin $i$ respectively. We finish the proof by noting that the set $\cal{B}'$ of bins occupied is a superset of the set $\cal{C}'$ of candidates who received proposals (the algorithm could end earlier than the balls-in-bins process, where we have unlimited balls). Also, for the above process, $\tau \leq \kappa$ and $Z \leq \hat{Z}$.
\end{proof}
As an immediate consequence, using symmetry, we have
\[
\mathbf{E}[Z] \leq \mathbf{E}[\tilde{Z}] = \mathbf{E}[\mathbf{E}[\tilde{Z}|\kappa]] = \mathbf{E}[\kappa/m].
\]
Since $m^*$  has $d$ women on his list, the expected highest rank proposal he receives is 
\[\mathbf{E}[Y] = \mathbf{E}\left[\mathbf{E}[Y|Z]\right] = \mathbf{E}[d/(Z+1)] \geq d/\mathbf{E}[Z+1] \geq d\cdot m / (\mathbf{E}[\kappa] + m) =: r_m.\]

So the expected rank of the highest ranked stable partner of every man is at least $r_m$. Therefore, the expected number of proposals in the MPDA algorithm is at least $m \cdot r_m = d\cdot m^2/\mathbf{E}[\kappa]$. The following proposition finishes the proof.

\begin{proposition}
$\mathbf{E}[\kappa] \leq (1 + o(1))m\ln(\frac{1 + \alpha}{\alpha + 1/m})$.
\end{proposition}

The proof is identical to Claim~\ref{lem:jpda-balls-bins} and so we omit it.
\end{proof}

\section{Conclusion}\label{sec:conclusion}
In this paper, we study a fundamental question about the existence of perfect stable matching in random matching markets with imbalance and partial preferences.  We prove sharp upper and lower bounds on the length of the preference list $d$ that results in a perfect stable matching.  

Partial preferences and imbalance are common occurrences in real matching markets.
We believe that the new insights from this paper about the interplay between them will serve as guiding principles in market design.

\section{Acknowledgements}
This paper was supported in part by NSF CCF 1947789.  

We would like to thank Jackson Ehrenworth and Max Enis for their contributions to a 
related project which led us to work on this paper.

\bibliographystyle{ACM-Reference-Format}
\bibliography{stablematching}
\section{Appendix: Omitted Proofs}\label{sec:appendix}
In this section, we include the proofs omitted from the main paper of the paper.

\pparagraph{Notation.} We use $[n]$ to denote $\{1, \ldots, n\}$ and $[a, b]$ to
denote $\{a, a+1, \ldots, b\}$ for integers $b\geq a$.

\subsection{Large deviation inequalities}
We recall the Chernoff bound (\cite{DP09} Theorem $1.1$). Let $X$ be a $\operatorname{Bin}(n,p)$ random variable. Then
\[
\mathbb{P}(|X - \mathbf{E}[X]| \geq \delta \mathbf{E}[X]) \leq 2\exp\left\{\delta^2\mathbf{E}[X]/3\right\}.
\]
We say a function $f:\mathcal{X}^n \rightarrow \mathbb{R}$ is $k$-Lipschitz if for every $x_1,\ldots,x_n \in \mathcal{X}$, we have
\[
\max_{x_i'\in \mathbf{X}}|f(x_1,\ldots,x_{i-1},x_i,x_{i+1},\ldots,x_n) - f(x_1,\ldots,x_{i-1},x_i',x_{i+1},\ldots,x_n)| \leq k.
\]
McDiarmid's inequality (\cite{DP09} Corollary $5.2$) states that for randomly and independently chosen $X_1,\ldots,X_n \in \mathcal{X}$, we have
\[
\mathbb{P}(|f(X_1,\ldots,X_n) - \mathbf{E}[f(X_1,\ldots,X_n)]| \geq t) \leq 2\exp\left\{-\frac{2t^2}{k^2n}\right\}
\]

\subsection{Sandwiching: Proof of Lemma~\ref{lem:sandwich}}

\begin{proof}[Proof of Lemma~\ref{lem:sandwich}]
For $d^{(\ell)} < d < d^{(u)}$, let $M_{\ell}$, $M$, and $M_u$ be random stable matching instances sampled from $M^C_{n,\alpha,d^{(\ell)}}$, $M_{n,\alpha,d}$, and $M^C_{n,\alpha, d^{(u)}}$. For a candidate $c$, let ordered subsets $L_{\ell}$, $L$, and $L_u$ of $[n]$ denote their neighborhood in $M_{\ell}$, $M$, and $M_u$ respectively.

We sample $M_{\ell}$, $M$, $M_{u}$ as follows: Each job has a total order on all the candidates, which are the same in all three instances. The ranking of a job $j$ is just the induced order on its neighbors. For each candidate $c$, sample $d_c$ from $\operatorname{Bin}(n,d/n)$, and set $d_{\max} := \max\{d_c,d^u\}$. Choose a tuple $(e_1,\ldots, e_{d_{\max}})$ where each $e_i \in [n]$ uniformly without repetitions. Set $N_L := (e_1,\ldots,d_L)$, $N = (e_1,\ldots,e_{d_c})$, and $N_{U}:=(e_1,\ldots,e_{d_U})$. Using Chernoff bound, we have that 

\[
\mathbb{P}(d_c\neq (1 \pm \delta)d) \leq \exp\left\{-\delta^2 d/3\right\}.
\]

So, with probability at least $1 - n(1+\alpha)e^{-\delta^2d/3}$, we have that $M_L \subseteq M \subseteq M_U$.

The proof follows in a similar fashion for $M^J$.
\end{proof}

\subsection{Probability Bounds for the Deferred Acceptance Algorithms: Proof of Claim~\ref{claim:probbounds}}

\begin{proof}[Proof of Claim~\ref{claim:probbounds}]

We compare the random proposals in the CPDA algorithm with a random balls-in-bins process using the following coupling.  In the CPDA algorithm, at each time step a new proposal is made by a candidate to a job that has not already rejected them and the time is measured as the number of such proposals.  

Consider a balls-in-bins process with $n$ bins, one for each job. Each time a random proposal for a candidate $c$ is required for the DA algorithms, balls are placed in randomly chosen bins until it lands in a bin that corresponds to a job that has not yet rejected $c$. This job constitutes a new proposal made by $c$ in the CPDA algorithm. At time $t$ in the balls-in-bins process, $t$ balls have been placed. Let $b_1^{t},\ldots, b_n^{t}$ be the number of  balls in the bins $,1 \ldots, n$ at time $t$. Let $j_1,\ldots,j_n$ be the number of proposals received by the jobs at time $\kappa$.

Let $T$ be the total number of balls placed in bins in order to obtain $\kappa$ proposals in the CPDA algorithm. Since, the expected number of balls placed to find one that can be used for the algorithm is at most $1 + d/n$, we have that $\mathbf{E}[T] \leq \kappa + O(d\ln n)$. Since $j_i \leq b_i$ by the coupling, Markov's inequality gives us that,

\begin{equation}\label{eqn:sumb_i}
\mathbb{P}\left(T \geq \kappa + \tilde{O}(\sqrt{n})\right) \leq 1/\sqrt{n}.
\end{equation}


To prove the part~(\ref{probb}) of Claim~\ref{claim:probbounds}, let $p_F^{\kappa} = 1 - \frac{1}{n}\sum_{i\in [n]}\frac{1}{b_i + 1}$ be the probability that the next proposal (at time $\kappa + 1$) is rejected. We will use 
\[
q_F^{\kappa} := \max_{t \in [\kappa, \kappa + \tilde{O}{\sqrt{n}}]}\left(1 - \frac{1}{n}\sum_{i\in [n]}\frac{1}{b_i^t + 1}\right)
\]
as a proxy for $p_F^{\kappa}$, since for every $t \geq \kappa$, we have that

\begin{align*}
0\leq \frac{1}{n}\sum_{i \in [n]}\left(\frac{1}{j_i + 1} - \frac{1}{b_i^t + 1}\right) & \leq \frac{1}{n}\sum_{i \in [n]}(b_i^t - j_i) = \frac{t - \kappa}{n}.
\end{align*}

Thus it suffices to prove that for $t \in [\kappa, \kappa + \tilde{O}{\sqrt{n}}]$, we have


\begin{equation}\label{eqn:p_F}
\mathbb{P}\left(\left|q_F^{t} - \left(1 - \frac{n}{n+\kappa}\right)\right| \geq 2\frac{n^2}{\kappa^2}\right) = 2\exp\left\{-\tilde{\Omega}(n)\right\}.
\end{equation}

Indeed, since Union bound gives us
\begin{align*}
&\mathbb{P}\left(\left|p_F^{\kappa} - \left(1 - \frac{n}{n+\kappa}\right)\right| \geq 4\frac{n^2}{\kappa^2}\right)\\
& \leq \mathbb{P}\left(\left|q_F^{T} - \left(1 - \frac{n}{n+\kappa}\right)\right| \geq 3\frac{n^2}{\kappa^2}\Big|T \leq \kappa + \tilde{O}(\sqrt{n})\right) + \mathbb{P}(T > \kappa + \tilde{O}(\sqrt{n})) \\
&  \leq \sum_{t = \kappa}^{\kappa + \tilde{O}(\sqrt{n})}\mathbb{P}\left(\left|q_F^{t} - \left(1 - \frac{n}{n+\kappa}\right)\right| \geq 2\frac{n^2}{\kappa^2}\Big|T \leq \kappa + \tilde{O}(\sqrt{n})\right) + \frac{1}{\sqrt{n}} \\
& \leq \frac{\sum_{t = \kappa}^{\kappa + \tilde{O}(\sqrt{n})}\mathbb{P}\left(\left|q_F^{t} - \left(1 - \frac{n}{n+\kappa}\right)\right| \geq 2\frac{n^2}{\kappa^2}\right)}{\mathbb{P}\left(T \leq \kappa + \tilde{O}(\sqrt{n})\right)} + \frac{1}{\sqrt{n}}\\
&= O\left( \frac{1}{\sqrt{n}}\right).
\end{align*}

Next, we obtain a bound on $\mathbf{E}[q_F^T]$.

\begin{proposition}\label{prop:Ep_F}
For $t \in [\kappa, \kappa + \tilde{O}{\sqrt{n}}]$, we have $\mathbf{E}[q_F] = 1 - \frac{n}{n+\kappa} \pm 2\left(\frac{n}{\kappa}\right)^2.$
\end{proposition}

\begin{proof}
Let $\mu :=  \mathbf{E}[1+b_i^t] = 1 + \frac{t}{n}$. By Jensen's inequality, we have 
\[
\mathbf{E}[1- q_F^t] = \mathbf{E}\left[\frac{1}{1+b_i^t}\right] \geq \frac{1}{\mathbf{E}[1+b_i^t]} = \frac{n}{n+t}.
\] 

For the other direction, we have:

\begin{align*}
\mathbf{E}[1-q_F^t] & = \mathbf{E}\left[\frac{1}{n}\sum_{i\in [n]}\frac{1}{b_i^t+1}\right] \\
& = \frac{1}{n}\sum_{i \in [n]}\mathbf{E}\left[\frac{1}{b_i^t + 1}\right] \\
& = \frac{1}{\mu\cdot n}\sum_{i \in [n]}\mathbf{E}\left[\frac{1}{1 + \frac{b_i^t + 1 - \mu}{\mu}}\right]\\
& \leq \frac{1}{\mu \cdot n}\sum_{i\in [n]}\mathbf{E}\left[1 - \left(\frac{b_i^t + 1 - \mu}{\mu}\right) + \left(\frac{b_i^t + 1 - \mu}{\mu}\right)^2\right] \\
& = \frac{n}{n+t} + \frac{n}{n+t}\mathbf{E}\left[\left(\frac{b_1^t + 1 - \mu}{\mu}\right)^2\right] \\
& \leq \frac{n}{n + t} + \left(\frac{n}{t}\right)^2\\
& \leq \frac{n}{n+\kappa} + 2\left(\frac{n}{\kappa}\right)^2.
\end{align*}

The second last inequality follows from the fact that $b_i^t$ is a $\operatorname{Bin}\left(t,\frac{1}{n}\right)$ random variable, so $\operatorname{var}(b_i^t) = \mathbf{E}\left[(b_i^t + 1 - \mu)^2\right] = \frac{t}{n}\left(1 - \frac{1}{n}\right)$.
\end{proof}

With Proposition~\ref{prop:Ep_F} in hand,~\eqref{eqn:p_F} follows from McDiarmid's inequality. Indeed, consider a sequence $\mathcal{X} = (\mathcal{X}_1,\ldots, \mathcal{X}_{t})$ of $t$ randomly chosen jobs, and define $b_i^t(\mathcal{X}) = |\left\{j\in [t]~|~\mathcal{X}_j = i\right\}|$. Let us denote 
\[
q(\mathcal{X}) = \frac{1}{n}\sum_{i \in [n]}\frac{1}{b_i^t(\mathcal{X}) + 1}.
\]

Consider two sequences of jobs $\mathcal{Y}$ and $\mathcal{Z}$ that match in all but the $i$'th proposal, i.e., $\mathcal{Y}_j = \mathcal{Z}_j$ for $j \neq i$. It is clear that $|q(\mathcal{Y}) - q(\mathcal{Z})| \leq \frac{2}{n}$. Thus $q_F(\cdot)$ is $\frac{2}{n}$-Lipschitz, and so after a random sequence of $t$ proposals, we have that for every $\alpha >0$,


\[
\mathbb{P}\left(|q_F^{\kappa} - \mathbf{E}q_F^{t}| > \alpha\right) \leq 2\exp\left\{-2\alpha^2 n^2/\kappa\right\}.
\]

It remains to plug in $\alpha = 3\frac{n^2}{\kappa^2}$, observe that $\kappa = \Theta(n \ln n)$, and union bound over all for $t \in [\kappa, \kappa + \tilde{O}{\sqrt{n}}]$ to finish the proof of~\eqref{eqn:p_F}. 

To prove part~(\ref{probe}) in Claim~\ref{claim:probbounds}, we observe that in the coupling described above, 
\[
\{i~|~j_i = 0\} \subseteq \{i~|~b_i^{\kappa} = 0\}.
\]
The expected number of empty bins at step $\kappa$ is $n(1-1/n)^t \leq ne^{-\kappa/n} = e^{C}$. So by Markov's inequality,

\[
\mathbb{P}(|\{i~|~j = 0\}|> e^{2C}) \leq \mathbb{P}(|\{i~|~b_i^T = 0\}|> e^{2C}) \leq e^{-C}.
\]

To prove part~(\ref{probl}) in Claim~\ref{claim:probbounds}, we have
\begin{align*}
\mathbb{P}(\{i~|~j_i = 0\} = \emptyset) & = \mathbb{P}(\{i~|~b_i^T = 0\} = \emptyset)\\
& \leq \mathbb{P}(\{i~|~b_i^T = 0\} = \emptyset|T<\kappa + \tilde{O}(\sqrt{n})) + \mathbb{P}(T \geq \kappa + \tilde{O}(\sqrt{n}))\\
& \leq \mathbb{P}\left(\{i~|~b_i^{\kappa + \tilde{O}(\sqrt{n})} = 0\} = \emptyset|T<\kappa + \tilde{O}(\sqrt{n})\right) +\frac{1}{\sqrt{n}}\\
& \leq \frac{\mathbb{P}\left(\{i~|~b_i^{\kappa + \tilde{O}(\sqrt{n})} = 0\} = \emptyset\right)}{\mathbb{P}(T < \kappa + \tilde{O}(\sqrt{n}))} + \frac{1}{\sqrt{n}} \\
& \leq 2\exp\left(-C/2\right) + \frac{1}{\sqrt{n}}.
\end{align*}
\end{proof}

\subsection{Expected number of rounds in JPDA: Proof of Claim~\ref{lem:jpda-balls-bins}}

\begin{proof}[Proof of Claim~\ref{lem:jpda-balls-bins}]
We again do this by coupling with balls-in-bins process, one bin for each candidate. We abuse notation and use $[m]$ to denote both the set of bins and the set of candidates. Consider the following coupling between the processes of JPDA and balls-in-bins. Let us denote $\tau_0 = (1+\gamma/2)m \cdot \ln \left(\frac{1+\alpha}{\alpha + 1/m}\right)$.

Suppose at a point in the JPDA algorithm, a job $j \in [n]$ is required to make a (random) proposal, and so far, $\mathbf{C}_j$ is the set of candidates that $j$ has already proposed to. We keep throwing balls-into-bins uniformly at random until we hit a bin $\mathbf{x} \not\in \mathbf{C}_j$. We assign $\mathbf{x}$ as the next proposal of $j$. Thus the distribution of the next proposal of $j$ is uniform on $[m] \setminus \mathbf{C}_j$. 

Let $\tau$ and $\tau'$ denote the number of proposals made and the number of balls thrown in the above algorithm, respectively. Clearly $\tau \leq \tau'$, and so $\mathbf{E}[\tau] \leq \mathbf{E}[\tau']$. Moreover, let $P$ and $P'$meat denote the candidates who have gotten at least one proposal and the bins that have gotten at least one ball just before the final throw. By the coupling, we have $P = P'$, and since there are at most $n$ matched candidates by the end, $|P| = |P'| \leq n-1$, since the JPDA algorithm terminates at the moment when $n$ candidates have a match. We have

\begin{align*}
\mathbb{P}\left(\tau > \tau_0\right) & \leq \mathbb{P}\left(\tau' >  \tau_0\right) \\
& \leq \mathbb{P}\left((|P'| < n)\land \tau' > \tau_0 \right) \\
& \leq \mathbb{P}\left(|P'| < n | \tau' >\tau_0 \right).
\end{align*}

We upper bound this probability by 

\[
\binom{m}{\alpha n + 1}\left(1 - \frac{\alpha n}{m}\right)^{\tau'} \leq \left(\frac{e(1 + \alpha)}{\alpha + 1/m}\right)^{\alpha n}\left(\frac{1}{1+\alpha}\right)^{\tau'}.
\]

Therefore,
\begin{align*}
\mathbf{E}[\tau] - \tau_0 & \leq \sum_{t > \tau_0}\mathbb{P}\left(\tau > t\right) \\
&\leq \left(\frac{e(1 + \alpha)}{\alpha + 1/m}\right)^{\alpha n}\sum_{t \geq \tau_0}\left(\frac{1}{1+\alpha}\right)^{t} \\
& \leq \left(\frac{e(1 + \alpha)}{\alpha + 1/m}\right)^{\alpha n + 1}\left(\frac{1}{1+\alpha}\right)^{\tau_0} \\
& \leq \left(\frac{e(1 + \alpha)}{\alpha + 1/m}\right)^{\alpha n + 1}e^{-\alpha \tau} \\
& \leq e^{\alpha n + 1}\left(\frac{\alpha + 1/m}{1 + \alpha}\right)^{\gamma/2 \alpha n} \\
& = o(1).
\end{align*}

\end{proof}

\end{document}